\documentclass[journal]{IEEEtran}

\usepackage{amsmath}
\usepackage{amsfonts}
\usepackage{amssymb}
\usepackage{amsthm}
\usepackage{tabularx}
\usepackage{mathrsfs} 

\usepackage{algorithm}
\usepackage{algorithmic}

\usepackage{url}

\newtheorem{Proposition}{Proposition}
\newtheorem{lemma}{Lemma}
\newtheorem{corollary}{Corollary}

\usepackage{stfloats}
\usepackage{float}
\usepackage{graphicx}
\usepackage{xcolor}
\usepackage{subfigure}
\usepackage[colorlinks=true, allcolors=blue]{hyperref}

\makeatletter
\def\blfootnote{\xdef\@thefnmark{}\@footnotetext}
\makeatother

\begin{document}
	
\title{\LARGE Performance Analysis of Wireless Communication Systems Assisted by Fluid Reconfigurable Intelligent Surfaces} 
\author{Farshad~Rostami~Ghadi,~\IEEEmembership{Member},~\textit{IEEE}, 
            Kai-Kit~Wong,~\IEEEmembership{Fellow},~\textit{IEEE}, 
            F. Javier~L\'opez-Mart\'inez,~\IEEEmembership{Senior~Member}, 
            George C. Alexandropoulos,~\IEEEmembership{Senior~Member,~IEEE}, and
            Chan-Byoung~Chae,~\IEEEmembership{Fellow},~\textit{IEEE}
            \vspace{-8mm}
	   }
\maketitle

\begin{abstract}
This letter investigates the performance of emerging wireless communication systems assisted by a fluid reconfigurable intelligent surface (FRIS). Unlike conventional reconfigurable intelligent surfaces (RISs), an FRIS consists of fluid-inspired metamaterials arranged in a densely packed matrix of sub-elements over a surface. It dynamically activates specific elements for signal reflection and modulation based on real-time channel conditions. Considering a downlink scenario where a base station communicates with a user terminal via a FRIS, we first characterize the statistical behavior of the equivalent end-to-end channel by deriving closed-form approximations for its cumulative distribution and probability density functions. Using these expressions, an analytical approximation for the outage probability and a tight upper bound on the ergodic capacity, including their asymptotic behaviors for high signal-to-noise ratio values, are derived. Our findings reveal key performance trends demonstrating that FRIS can substantially improve link reliability and spectral efficiency compared to conventional RISs, owing to its capability to dynamically select optimal elements from a dense preconfigured grid.
\end{abstract}

\maketitle

\blfootnote{The work of F. Rostami Ghadi is supported by the European Union's Horizon 2022 Research and Innovation Programme under Marie Sk\l odowska-Curie Grant No. 101107993. The work of K. K. Wong is supported by the Engineering and Physical Sciences Research Council (EPSRC) under Grant EP/W026813/1. The work of F. J. L\'opez-Mart\'inez is funded by MICIU/AEI/10.13039/50110001103 and FEDER/UE through grant PID2023-149975OB-I00 (COSTUME). The work of C.-B. Chae is supported by by IITP (2025-RS-2024-00428780,
2022R1A5A1027646).} 
\blfootnote{\noindent F. Rostami Ghadi and F. J. L\'opez-Mart\'inez are with the Department of Signal Theory, Networking and Communications, Research Centre for Information and Communication Technologies (CITIC-UGR), University of Granada, 18071, Granada, Spain (e-mail: $\rm {f.rostami, fjlm}@ugr.es)$.}
\blfootnote{\noindent K. K. Wong is affiliated with the Department of Electronic and Electrical Engineering, University College London, Torrington Place, WC1E 7JE, U.K. and he is also affiliated with Yonsei Frontier Lab, Yonsei University, Seoul, Korea (e-mail: $\rm kai\text{-}kit.wong@ucl.ac.uk$).}
\blfootnote{G. C. Alexandropoulos is with the Department of Informatics and Telecommunications, National and Kapodistrian University of Athens, 16122 Athens, Greece and the Department of Electrical and Computer Engineering, University of Illinois Chicago, IL 60601, USA (e-mail: $\rm alexandg@di.uoa.gr$).}
\blfootnote{C.-B. Chae is with the School of Integrated Technology, Yonsei University, Seoul, 03722 South Korea (e-mail: $\rm cbchae@yonsei.ac.kr$).}

\blfootnote{Corresponding Author: Kai-Kit Wong.}

\vspace{-2mm}
\section{Introduction}\label{sec-intro}
\IEEEPARstart{R}{ecently}, reconfigurable intelligent surface (RIS) technologies have emerged as an enabler for next generation wireless communications by intelligently controlling the reflection of radio signals~\cite{Basar-2024}. Through a large array of passive reflecting elements, RISs can improve coverage, energy efficiency, and spectrum usage without requiring active transmission \cite{Alexandropoulos-2021}. Nonetheless, practical implementations of RISs face several challenges~\cite{wei2021channel}, such as complex channel estimation as the number of elements grows, and their passive nature which is subject to the multiplicative pathloss effect, often resulting in severe signal attenuation. Additionally, current RIS designs are geometrically rigid, this eliminates any possibility of fully utilizing the spatial diversity when physical space in some cases is abundant.

To empower RISs for greater spatial adaptivity, it is possible to leverage the fluid antenna system (FAS) concept, proposed by Wong {\em et al.}~\cite{wong2021fluid} for shape and position reconfigurability in antennas. In recent years, FAS has been shown to be effective for many applications \cite{new2024tut}. The synergy between RIS and FAS has been also recently addressed in \cite{ghadi2024on,ghadi2025secrecy}. Inspired by this concept, the notion of a fluid RIS (FRIS) has been introduced, which proposes that the RIS elements can also dynamically reconfigure their spatial positions within a predefined area \cite{salem2025first,xiao2025fluid}. This physical adaptivity enables FRIS to obtain more spatial diversity and enhance link robustness in dynamic wireless environments. Moreover, a fluid integrated reflecting and emitting surface (FIRES) has also emerged expanding FRIS application for full $360^\circ$ coverage \cite{ghadi2025fires}. Despite these early results, however, the performance analysis and optimization of FRIS-based systems is still in its infancy. 

Motivated by this gap, our aim in this letter is to provide a performance analysis of FRIS-based systems. Specifically, focusing on the optimized FRIS system proposed in \cite{xiao2025fluid}, we derive closed-form expressions for the probability density function (PDF) and cumulative distribution function (CDF) of the equivalent channel gain by using moment-matching techniques. We also analyze the outage probability, ergodic capacity, and the corresponding asymptotic behaviors in the high-signal-to-noise ratio (SNR) regime, providing valuable insights into the reliability and spectral efficiency of FRIS systems. Our numerical results highlight significant performance trends, demonstrating that FRIS offers a substantial improvement in link reliability and spectral efficiency over conventional RISs. Notably, the ability of FRIS to dynamically select optimal elements from a dense preset grid provides enhanced adaptability, enabling better exploitation of spatial diversity, thus offering more robust communications in dynamically varying environmental conditions.

\begin{figure}[!t]
\centering
\includegraphics[width=1.0\columnwidth]{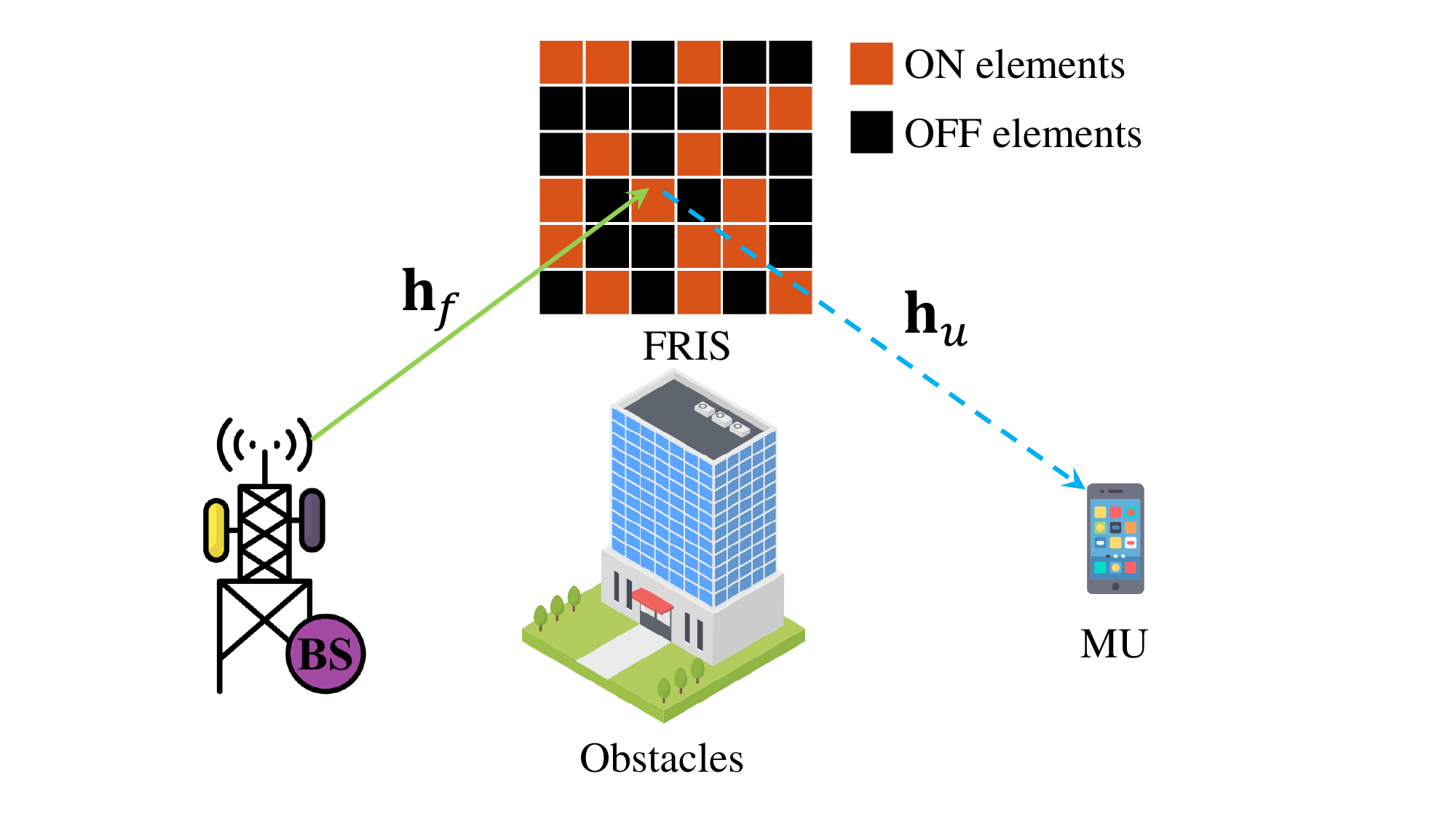}\vspace{-3mm}
\caption{The considered FRIS-assisted wireless communication system with the FRIS comprising ON-OFF unit elements.}\label{fig_model}
\vspace{-4mm}
\end{figure}

\vspace{-2mm}
\section{System  Model}\label{sec-sys}
We consider a wireless communications setup, as illustrated in Fig.~\ref{fig_model}, in which a base station (BS) with a fixed-position antenna (FPA) aims to send information to a FPA mobile user (MU) with the help of an FRIS. It is assumed that the direct link between the BS and MU is blocked due to obstacles. The FRIS includes $M=M_x\times M_z$ reflecting elements, operating in one of two discrete modes ON or OFF. When set to the ON mode, the element engages with the incoming wave, modifying its electromagnetic behavior to support the desired system functionality. In the OFF mode, the element is terminated with a matched load, effectively isolating it from the incident wave and preventing any signal alteration. The FRIS elements are uniformly distributed over a surface with size of $W=W_x\lambda\times W_z\lambda$, where $\lambda$ denotes the carrier wavelength. Specifically, $M_x$ and $W_x$ are the number of elements in each row with total length of $W_x\lambda$ and $M_z$, and $W_z$ are the number of elements in each column with total length of $W_x\lambda$. Very close placement of adjacent elements is assumed and, as a consequence, spatial correlation must be taken into account. To this end, we deploy the Jakes' model to describe the correlation coefficient between any two arbitrary FRIS elements $i$ and $j$ under rich scattering conditions: $j_{i,j} = \mathcal{J}_0\left(\frac{2\pi d_{i,j}}{\lambda}\right)$, where $\mathcal{J}_0\left(\cdot\right)$ denotes the zero-order spherical Bessel function of the first kind and $d_{ij}$\footnote{The derivation of squared distances between FRIS elements depends on how the two-dimensional (2D) indices are mapped from the one-dimensional (1D) element index. If a row-major (i.e., row-wise) indexing is used, coordinates are typically extracted using $x = \text{mod}(i-1, M_x)$ and $z = \lfloor (i-1) / M_x \rfloor$. Alternatively, if column-major (column-wise) indexing is used, the expressions become: $x = \lfloor (i-1) / M_z \rfloor$ and $z = \text{mod}(i-1, M_z)$. As long as the coordinate extraction is consistent with the chosen indexing scheme, the resulting spatial distances, thus the spatial correlation matrix, remain invariant.} is the physical distance between elements $i$ and $j$, given by
\begin{multline}
d_{i,j} = d_x^2\left(\mathrm{mod}\left(i,M_x\right)-\mathrm{mod}\left(j,M_x\right)\right)^2\\
+d_z^2\left(\left\lfloor\frac{i}{M_x}\right\rfloor-\left\lfloor\frac{j}{M_x}\right\rfloor \right),
\end{multline}
where $d_x=\frac{W_x\lambda}{M_x}$ and $d_z=\frac{W_z\lambda}{M_z}$ are the physical inter-element spacing in each FRIS row and column, respectively. Thus, the spatial correlation matrix is formulated as follows:
\begin{align}
\mathbf{J} = \begin{bmatrix}
	j_{1,1} & j_{1,2} & \cdots & j_{1,M} \\
	j_{2,1} & j_{2,2} & \cdots & j_{2,M} \\
	\vdots & \vdots & \ddots & \vdots \\
	j_{M,1} & j_{M,2} & \cdots & j_{M,M} \\
\end{bmatrix}.
\end{align}

Assuming that $M_o$ elements will be turned ON to modulate the incident signal, the received signal at the MU can be mathematically expressed as follows:
\begin{align}
y = \sqrt{PL_f L_u}\underset{H_\mathrm{eq}}{\underbrace{\mathbf{h}_{u}^H\mathbf{J}^{\frac{1}{2}}\mathbf{S}^T_{M_{o}}\mathbf{\Phi}\mathbf{S}_{M_{o}}\mathbf{J}^{\frac{1}{2}}\mathbf{h}_{f}}}x+z,\label{eq-y}
\end{align}
where $L_k=\sqrt{\rho d_k^{-\alpha}}$, for $k\in\left\{f,u\right\}$, represents the large-scale path-loss in which $\rho$ is the reference gain at one meter, $\alpha$ is the path-loss exponent, and $d_k$ represents the distances of the BS-to-FRIS and FRIS-to-MU links. Furthermore, $P$ is the power transmitted by the BS, $x$ is the transmitted information signal, and $z$ is the additional white Gaussian noise (AWGN) with zero mean and variance $\sigma^2$, i.e., $z\sim\mathcal{CN}\left(0,\sigma^2\right)$. The term $H_\mathrm{eq}$ is the equivalent fading channel in which the  vector $\mathbf{h}_k\in\mathbb{C}^{M\times 1}$ includes entries that are independent and identically distributed (i.i.d.) circularly symmetric complex Gaussian random variables (RVs) with zero mean and unit variance, i.e., $\mathbf{h}_k\sim\mathcal{CN}\left(0,\mathbf{I}_M\right)$,  representing the small-scale fading for the BS-to-FRIS or FRIS-to-MU links. The diagonal matrix $\mathbf{\Phi}=\mathrm{diag}\left(\left[\mathrm{e}^{j\phi_1}, \dots, \mathrm{e}^{j\phi_M}\right]\right)\in\mathbb{C}^{M\times M}$ contains the adjustable phases of the reflecting elements of the FRIS and $\mathbf{S}^T_{M_o}=\left[\mathbf{s}_1,\dots,\mathbf{s}_{M_o}\right]\in\mathbb{R}^{M\times M_o}$ is the element selection matrix, in which $\mathbf{s}_{m_o}$, for $m_o\in\mathcal{M}_o=\left\{1,\dots,M_o\right\}$, indicates one of the columns of the $M\times M$ identity matrix. 

Note that in fluid-inspired RIS systems, only a subset of the available meta-elements is actively utilized for reflection at any given time, based on instantaneous channel conditions or optimization objectives. The selection matrix $\mathbf{S}_{M_o}$ captures this sparse activation by projecting the full FRIS aperture of size $M$ onto a subspace of $M_o$ activated (i.e., ON) elements. This construction inherently gives $\mathbf{S}_{M_o}$ the structure of a binary orthogonal projection matrix, satisfying $\mathbf{S}_{M_o}^T\mathbf{S}_{M_o}=\mathbf{I}_{M_o}$, and accurately models the dynamic and selective behavior of FRIS, while enabling tractable analytical formulations. It is noteworthy to mention that the phase shifts in $\mathbf{\Phi}$ are assumed ideal and fully controllable to coherently align the reflected signals. The selection matrix $\mathbf{S}_{M_o}$ models the dynamic activation of the $M_o$ elements that maximize the equivalent channel gain, capturing the fluid antenna concept and indicating an optimized system configuration.

\vspace{-2mm}
\section{Performance Analysis}
In this section, we first derive the CDF and PDF of the equivalent channel gain at the MU, which are then used for obtaining novel closed-from approximations of the outage probability and ergodic capacity.

\subsection{Statistical Characterization}
Given \eqref{eq-y}, the received SNR at the MU is defined as 
\begin{align}
\gamma=\overline{\gamma}L_fL_u\left|H_\mathrm{eq}\right|^2,\label{eq-snr}
\end{align}
where $\overline{\gamma}\triangleq P/\sigma^2$ is the transmit SNR and $G\triangleq\left|H_\mathrm{eq}\right|^2= \left|\mathbf{h}_{u}^H\mathbf{J}^{\frac{1}{2}}\mathbf{S}^T_{M_{o}}\mathbf{\Phi}\mathbf{S}_{M_{o}}\mathbf{J}^{\frac{1}{2}}\mathbf{h}_{f}\right|^2$ represents the equivalent channel power gain. It is evident that $G$ involves the sum of products of several RVs, and its distribution is affected by the channel $\mathbf{h}_k$ and the correlation matrix $\mathbf{J}$. Also, the selection and reflection matrices $\mathrm{S}_{M_o}$ do not fundamentally affect the statistical distribution of the power gain, only as a scaling factor.

\begin{lemma}
The equivalent channel gain $G=\left|H_\mathrm{eq}\right|^2$ follows a generalized chi-squared distribution.
\end{lemma} 

\begin{proof}
We commence with the definition of $\mathbf{A} \triangleq \mathbf{J}^{\frac{1}{2}}\mathbf{S}^T_{M_o}\mathbf{\Phi}\mathbf{S}_{M_o}\mathbf{J}^{\frac{1}{2}}$, yielding the compact formulation $H_\mathrm{eq} = \mathbf{h}_u^H\mathbf{A}\mathbf{h}_f$. Next, we define $\mathbf{h}_f$'s linear transformation $\mathbf{g} \triangleq \mathbf{A}\mathbf{h}_f\in\mathbb{C}^{M\times 1}$, for which it holds $\mathbf{g}\sim\mathcal{CN}\left(0,\mathbf{A}\mathbf{A}^H\right)$, since $\mathbf{A}$ is deterministic for a given optimized configuration of phase shifts and element selection. Consequently, since $\mathbf{h}_u$ and $\mathbf{g}$ are independent, $H_\mathrm{eq}$ becomes a quadratic form in complex Gaussian RVs. Then, by defining the vector $\mathbf{h}^T\triangleq\left[\mathbf{h}_u \mathbf{h}_f\right]\sim\mathcal{CN}\left(0,\mathbf{I}_{2M}\right)$, we can write
$
G = \left|\mathbf{h}^H\begin{bmatrix}
	\mathbf{0} & \mathbf{I} \\
	\mathbf{0} & \mathbf{0} 
\end{bmatrix}\mathbf{h}\right|^2.
$ 
Alternatively, using the definition $\mathbf{B}\triangleq\mathbf{A}^H\mathbf{h}_u\mathbf{h}^H_u\mathbf{A}$, yields:
\begin{align}
G = \mathbf{h}_f^H\mathbf{B}\mathbf{h}_f. \label{eq-heq}
\end{align}
Consequently, given that $\mathbf{h}_f\sim\mathcal{CN}\left(0,\mathbf{I}_M\right)$ and $\mathbf{B}\succeq 0$, i.e., rank-one positive semi-definite (PSD), it can be concluded that $G$ is a quadratic form in a complex Gaussian vector having with a Hermitian PSD coefficient matrix. It is known from~\cite{am1992} that a Hermitian quadratic form in a complex Gaussian random vector of the form $Q=\mathbf{w}^H\mathbf{H}\mathbf{w}$ with $\mathbf{w\sim\mathcal{CN}\left(0,\mathbf{I}\right)}$ and $\mathbf{H}\succeq 0$ follows a generalized-$\chi^2$ distribution, expressible as a weighted sum of independent exponential RVs. Hence, by comparing $Q$ with the equivalent channel gain $G$ in \eqref{eq-heq}, $H_\mathrm{eq}$ is generalized-$\chi^2$-distributed, i.e., 
$
H_\mathrm{eq} \sim \sum_{l=1}^r \zeta_l\left|\omega_l\right|^2
$, 
where $\zeta_l$'s are the non-zero eigenvalues of $\mathbf{B}$ and $w_l\sim\mathcal{CN}\left(0,1\right)$ $\forall l=1,2,\ldots,r$.
\end{proof}

\begin{Proposition}\label{pro-1}
The distribution of the equivalent channel gain $G$ can be accurately approximated by the Gamma distribution via moment matching method, i.e., its CDF and PDF are given respectively by:
\begin{align}
F_G\left(g\right) = \frac{1}{\Gamma\left(k\right)}\Upsilon\left(k,\frac{g}{\theta}\right)\label{eq-cdf}
\end{align}
and 
\begin{align}
f_G\left(g\right) = \frac{1}{\theta^k\Gamma\left(k\right)}g^{k-1}\mathrm{e}^{-g/\theta},\label{eq-pdf}
\end{align}
where $k =\frac{\left(\mathrm{tr}\left(\widetilde{\mathbf{J}}^2\right)\right)^2}{\mathrm{tr}\left(\widetilde{\mathbf{J}}^4\right)}$, $\theta = \frac{\mathrm{tr}\left(\widetilde{\mathbf{J}}^4\right)}{\mathrm{tr}\left(\widetilde{\mathbf{J}}^2\right)}$, and $\widetilde{\mathbf{J}}=\mathbf{S}_{M_o} \mathbf{J} \mathbf{S}_{M_o}^T$.
\end{Proposition}

\begin{proof}
We begin by expanding $G$ as follows:
\begin{align}
G = \left|\mathbf{h}_u^H\mathbf{A}\mathbf{h}_f\right|^2= \mathbf{h}_u^H\mathbf{A}\mathbf{h}_f\mathbf{h}^H_f\mathbf{A}^H\mathbf{h}_u,
\end{align}
where $\mathbf{A}$ has been defined in Lemma 1. Then, the first moment of $G$ can be derived as:
\begin{align}
  \mathbb{E}\left[G\right]& = \mathbb{E}_{\mathbf{h}_u}\left[\mathbf{h}_u^H\mathbf{A}\mathbb{E}_{\mathbf{h}_f}\left[\mathbf{h}_f\mathbf{h}_f^H\right]\mathbf{A}^H\mathbf{h}_u\right] \\
&= \mathbb{E}_{\mathbf{h}_u}\left[\mathbf{h}_u^H\mathbf{A}\mathbf{A}^H\mathbf{h}_u\right]\overset{(a)}{=} \mathrm{tr}\left(\mathbf{A}\mathbf{A}^H\right), 
\end{align}
where $(a)$ follows from the fact tha $\mathbb{E}\left[ \mathbf{w}^H\mathbf{H}\mathbf{w}\right]=\mathrm{tr}\left(\mathbf{H}\right)$ for $\mathbf{w\sim\mathcal{CN}\left(0,\mathbf{I}\right)}$.

For the second moment of $G$, we have that:
\begin{multline}
\mathbb{E}\left[G^2\right]\\
=\mathbb{E}\left[\left(\mathbf{h}_u^H\mathbf{A}\mathbf{h}_f\right)\left(\mathbf{h}^H_f\mathbf{A}^H\mathbf{h}_u\right)\left(\mathbf{h}_u^H\mathbf{A}\mathbf{h}_f\right)\left(\mathbf{h}^H_f\mathbf{A}^H\mathbf{h}_u\right)\right]. 
\end{multline}
After some simplifications, and using Isserlis' (Wick's) theorem for complex Gaussian vectors, yields:
\begin{align}
\mathbb{E}\left[G^2\right] = \mathrm{tr}^2\left(\mathbf{A}\mathbf{A}^H\right)+\mathrm{tr}\left(\left(\mathbf{A}\mathbf{A}^H\right)^2\right), 
\end{align}
and as a result, the variance of $G$ is given by $\mathrm{var}\left(G\right) = \mathrm{tr}\left(\left(\mathbf{A}\mathbf{A}^H\right)^2\right)$.
Then, by applying the moment matching technique to the Gamma distribution with the mean of $k\theta$ and variance of  $k\theta^2$, in which $k$ and  $\theta$ are the shape and scale parameters, respectively, i.e., $G\sim\mathrm{Gamma}\left(k, \theta\right)$, results in:
\begin{align}\label{eq-kt1}
k = \frac{\left(\mathrm{tr}\left(\mathbf{A}\mathbf{A}^H\right)\right)^2}{\mathrm{tr}\left(\left(\mathbf{A}\mathbf{A}^H\right)^2\right)}~\mbox{and}~
\theta = \frac{\mathrm{tr}\left(\left(\mathbf{A}\mathbf{A}^H\right)^2\right)}{\mathrm{tr}\left(\mathbf{A}\mathbf{A}^H\right)}.
\end{align} 
To complete the proof, we need to further evaluate $\mathrm{tr}\left(\mathbf{A}\mathbf{A}^H\right)$ and $\mathrm{tr}\left(\left(\mathbf{A}\mathbf{A}^H\right)^2\right)$. Note first that:
\begin{align}
\mathbf{A}\mathbf{A}^H = \mathbf{J}^{\frac{1}{2}} \mathbf{S}_{M_o}^T \mathbf{\Phi} \mathbf{S}_{M_o} \mathbf{J} \mathbf{S}_{M_o}^T \mathbf{\Phi}^H \mathbf{S}_{M_o} \mathbf{J}^{\frac{1}{2}}. \label{eq-aah-updated}
\end{align}
Let us define the reduced-size matrix $\widetilde{\mathbf{J}}\triangleq\mathbf{S}_{M_o} \mathbf{J} \mathbf{S}_{M_o}^T\in\mathbb{C}^{M_o\times M_o}$, which is the principal submatrix of $\mathbf{J}$ containing rows and columns corresponding to the selected $M_o$ elements. We also define the reduced-size matrix $\tilde{\mathbf{\Phi}}=\mathbf{S}_{M_o}^T \mathbf{\Phi}\mathbf{S}_{M_o}\in\mathbb{C}^{M_o\times M_o}$, which is a diagonal matrix containing the selected phase shifts. Using the latter definitions, the trace of $\mathbf{AA}^H$ can simplified as follows:
\begin{align}
\mathrm{tr}\left(\mathbf{A}\mathbf{A}^H\right)
&= \mathrm{tr}\left(\mathbf{J}^{\frac{1}{2}} \mathbf{S}_{M_o}^T \mathbf{\Phi} \mathbf{S}_{M_o} \mathbf{J} \mathbf{S}_{M_o}^T \mathbf{\Phi}^H \mathbf{S}_{M_o} \mathbf{J}^{\frac{1}{2}}\right)\\
&\overset{(b)}{=} \mathrm{tr}\left(\tilde{\mathbf{\Phi}} \tilde{\mathbf{J}} \tilde{\mathbf{\Phi}}^H \tilde{\mathbf{J}} \right), \label{eq-trace1-step}
\end{align}
where $(b)$ follows from the cyclic property of the trace and the orthonormality of $\mathbf{S}_{M_o}$, i.e., $\mathbf{S}_{M_o}^T\mathbf{S}_{M_o}=\mathbf{I}_{M_o}$.

Since $\tilde{\mathbf{\Phi}}$ is a diagonal matrix with unit-modulus complex entries, i.e., phase rotations, 
$\tilde{\mathbf{\Phi}} \tilde{\mathbf{J}} \tilde{\mathbf{\Phi}}^H=\widetilde{\mathbf{J}}'$ holds with $|\widetilde{\mathbf{J}}'_{i,j}|=|\widetilde{\mathbf{J}}_{i,j}|$, yielding to the following expression:
\begin{align}
\mathrm{tr}\left(\mathbf{A}\mathbf{A}^H\right) = \mathrm{tr}\left(\tilde{\mathbf{J}}^2\right). \label{eq-trac1}
\end{align}

As for the $\mathrm{tr}\left(\left(\mathbf{A}\mathbf{A}^H\right)^2\right)$, using the expansion of $\mathbf{AA}^H$ from \eqref{eq-aah-updated}, then applying the cyclic property of trace, and simplifying with $\widetilde{\mathbf{J}}$ and $\tilde{\mathbf{\Phi}}$, results in:
\begin{align}
\mathrm{tr}\left((\mathbf{A}\mathbf{A}^H)^2\right) &= \mathrm{tr}\left((\tilde{\mathbf{\Phi}} \tilde{\mathbf{J}} \tilde{\mathbf{\Phi}}^H)^2 \tilde{\mathbf{J}}^2\right). 
\end{align}
Similarly, since $\tilde{\mathbf{\Phi}}$ is diagonal and unit-modulus, the matrix similarity transformation preserves the Frobenius norm. Hence,
\begin{align}
\mathrm{tr}\left((\mathbf{A}\mathbf{A}^H)^2\right) = \mathrm{tr}\left(\tilde{\mathbf{J}}^4\right). \label{eq-trac2}
\end{align}
Eventually, by substituting \eqref{eq-trac1} and \eqref{eq-trac2} into \eqref{eq-kt1}, the desired result is obtained, which completes the proof. 
\end{proof}

To verify the accuracy of the theoretical analysis, we compare the theoretical and empirical results of the equivalent channel gain $G$ in terms of its PDF and CDF in Fig.~\ref{fig_dist} for a system configuration with $M_o=144$ activated elements out of $M=\left(M_x,M_z\right)=\left(20,20\right)$. As shown, the theoretical curves closely match the simulation results, confirming the validity of the derived expressions. 

\begin{figure}[htbp]
\centering
\subfigure[]{\includegraphics[width=0.5\linewidth]{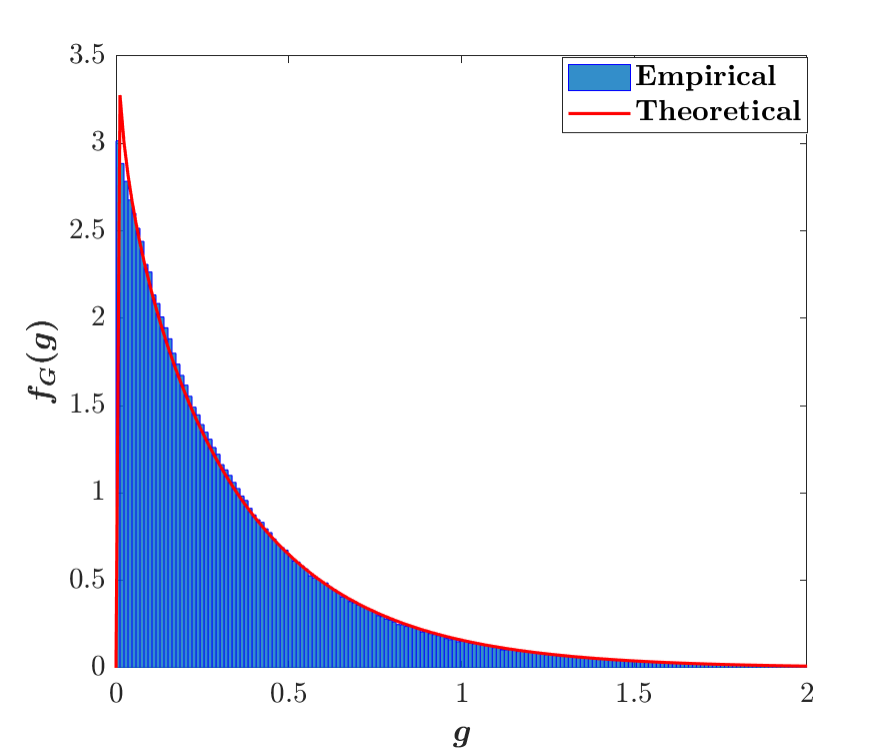}\label{fig:1a}
	}\hspace{-0.9cm}
	\hfill \vspace{-0.3cm}
\subfigure[]{\includegraphics[width=0.5\linewidth]{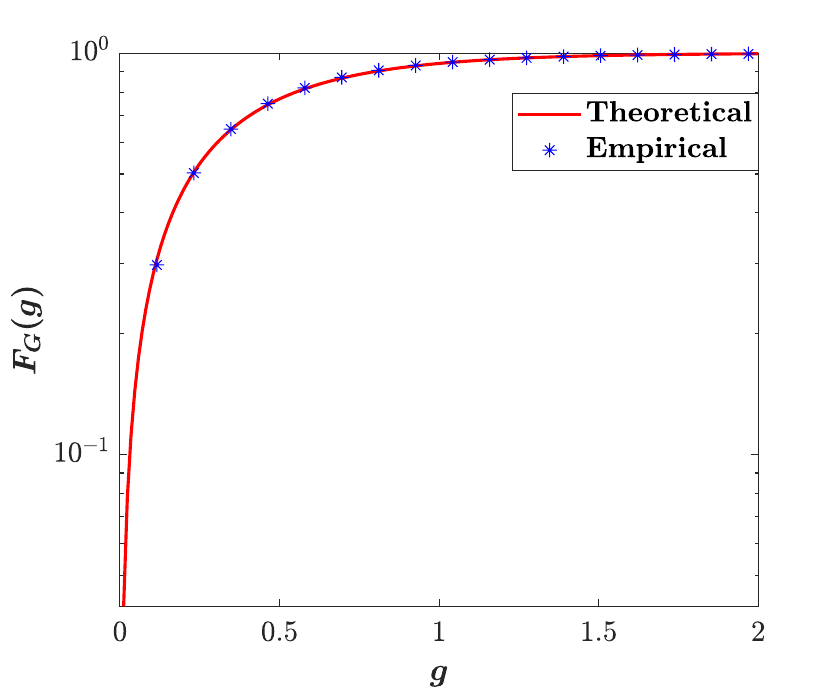}\label{fig:1b}
	}
	\hfill
\subfigure[]{\hspace{1cm}\includegraphics[width=0.65\linewidth]{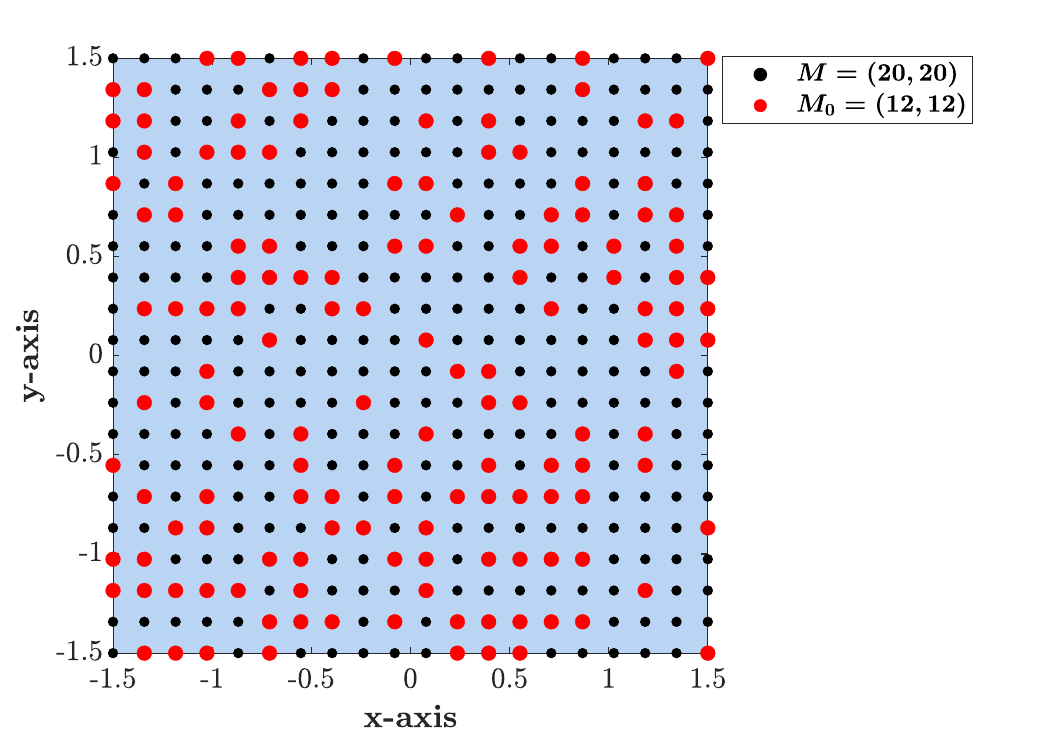}\label{fig_point}
	}\vspace{-2mm}
\caption{Illustration of (a) the PDF and (b) the CDF of the equivalent channel gain $G$, and (c) the system configuration for $M=\left(M_x,M_z\right)=\left(20,20\right)$.}\label{fig_dist}
\vspace{-6mm}
\end{figure}

\subsection{Outage Probability}
Outage probability is defined as the probability that the instantaneous channel capacity $C\triangleq\log_2\left(1+\gamma\right)$ falls below a target rate $R$, i.e., $P_o \triangleq \Pr\left(C\leq R\right)$. 

\begin{Proposition}
The outage probability for the proposed FRIS system is given by
\begin{align}
P_o  = \frac{1}{\Gamma\left(k\right)}\Upsilon\left(k,\frac{\overline{R}}{ \overline{\gamma}L_fL_u\theta}\right),
\end{align}
where $k$ and $\theta$ are defined in Proposition \ref{pro-1}.
\end{Proposition}

\begin{proof}
By inserting the SNR from \eqref{eq-snr} into the definition of outage probability, yields:
\begin{align}
P_o & = \Pr\left(\log_2\left(1+\overline{\gamma}L_fL_uG\right)\leq R\right)\\
& = \Pr\left(G\leq \frac{\overline{R}}{\overline{\gamma}L_fL_u}\right)=F_G\left(\frac{\overline{R}}{\overline{\gamma}L_fL_u}\right),\label{eq-p1}
\end{align}
where $\overline{R}\triangleq 2^{R}-1$. By inserting \eqref{eq-cdf} into \eqref{eq-p1}, the proof is completed. 
\end{proof}

\begin{corollary}
The asymptotic expression of the outage probability in the high-SNR regime, i.e., when $\gamma\gg 1$, is given by $P_o^\infty = \frac{1}{k\Gamma\left(k\right)}\left(\frac{\overline{R}}{ \overline{\gamma}L_fL_u\theta}\right)^k$.
\end{corollary}

\begin{proof}
In the high SNR regime, the CDF in \eqref{eq-cdf} can be rewritten as $F^\infty_G\left(g\right) \simeq \frac{1}{k\Gamma\left(k\right)}\left(\frac{g}{\theta}\right)^k$. Now, by substituting $F^\infty_G\left(g\right)$ into the outage probability definition, $P_o^\infty$ is obtained. 
\end{proof}

\subsection{Ergodic Capacity}
The ergodic capacity of the considered FRIS-assisted communication system is defined as follows: 
\begin{align}
\overline{C} \triangleq \mathbb{E}\left[C\right]=\int_0^\infty \log_2\left(1+\overline{\gamma}L_fL_uG\right)f_G\left(g\right)\mathrm{d}g,\label{eq-ecdef}
\end{align}
where $f_G\left(g\right)$ has been defined in \eqref{eq-pdf}. By substituting \eqref{eq-pdf} into \eqref{eq-ecdef}, a closed-form expression for the ergodic capacity can be derived in terms of the hypergeometric/Meijer's G-functions~\cite{Meijer}. However, to enhance the analysis further, we apply Jensen's inequality to derive a tight upper bound for $\overline{C}$. 

\begin{Proposition}
The ergodic capacity for the proposed FRIS-assisted communication system is given by 
\begin{align}
\overline{C} \approx \log_2\left(1+\overline{\gamma}L_f L_u \mathrm{tr}\left(\widetilde{\mathbf{J}}^2\right)\right).\label{eq-ec}
\end{align}
\end{Proposition}

\begin{proof}
By applying Jensen's inequality into \eqref{eq-ecdef}, we have
\begin{align}
\overline{C}&=\mathbb{E}\left[\log_2\left(1+\overline{\gamma}L_fL_uG\right)\right]\\
&\leq \log_2\left(1+\overline{\gamma}L_fL_u\mathbb{E}\left[G\right]\right)\\
&= \log_2\left(1+\overline{\gamma}L_fL_u\mathrm{tr}\left(\mathbf{A}\mathbf{A}^H\right)\right).\label{eq-p2}
\end{align}
By inserting $\mathrm{tr}\left(\mathbf{A}\mathbf{A}^H\right)$ from \eqref{eq-trac1} into \eqref{eq-p2}, yields\eqref{eq-ec} which completes the proof.
\end{proof}

\begin{corollary}
The asymptotic expression of the ergodic capacity in the high-SNR regime, i.e., $\gamma\gg 1$, is given by $\overline{C}^\mathrm{\infty} \approx \log_2\left(\overline{\gamma}L_f L_u \mathrm{tr}\left(\widetilde{\mathbf{J}}^2\right)\right)$.
\end{corollary}

\begin{figure*}[htbp]
\centering
\subfigure[]{\includegraphics[width=0.32\linewidth]{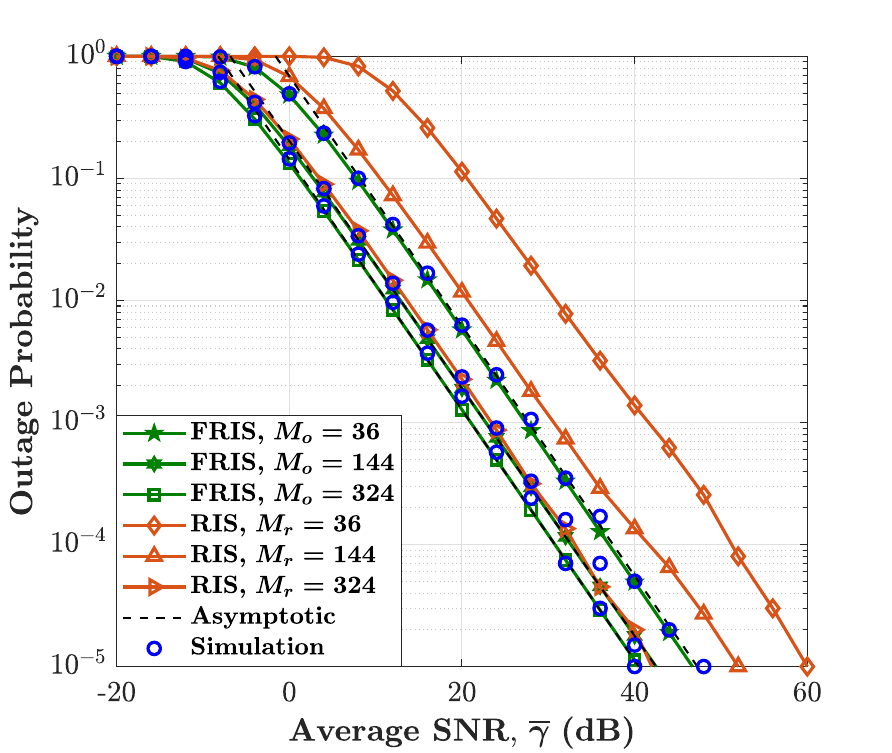}\label{fig_op}
	}\hspace{-2mm}
	\hfill
\subfigure[]{\includegraphics[width=0.32\linewidth]{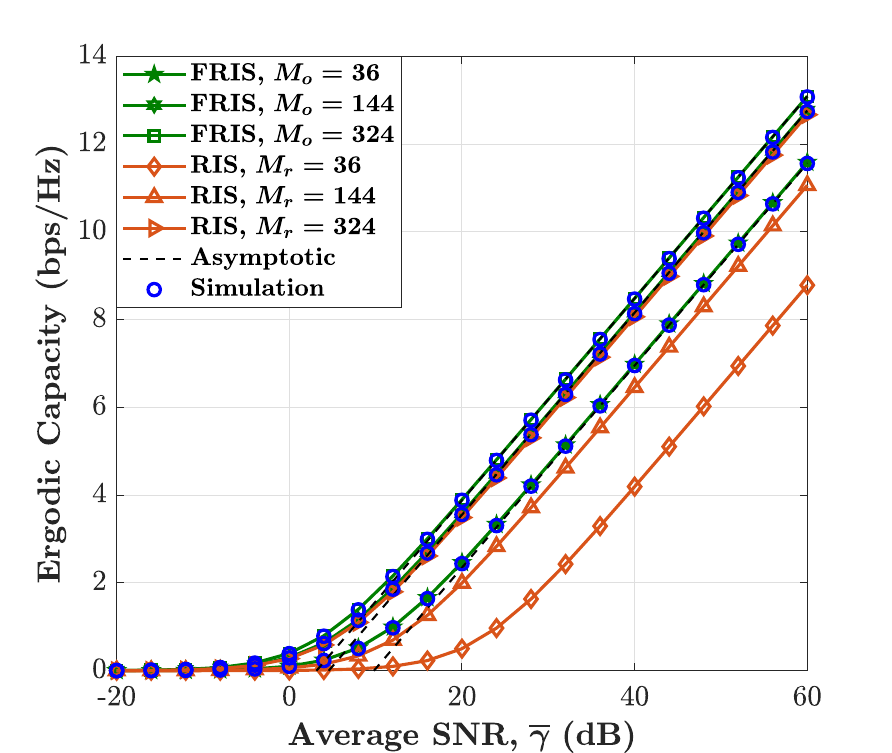}\label{fig_ec}
	}\hspace{-2mm}
	\hfill
\subfigure[]{\includegraphics[width=0.32\linewidth]{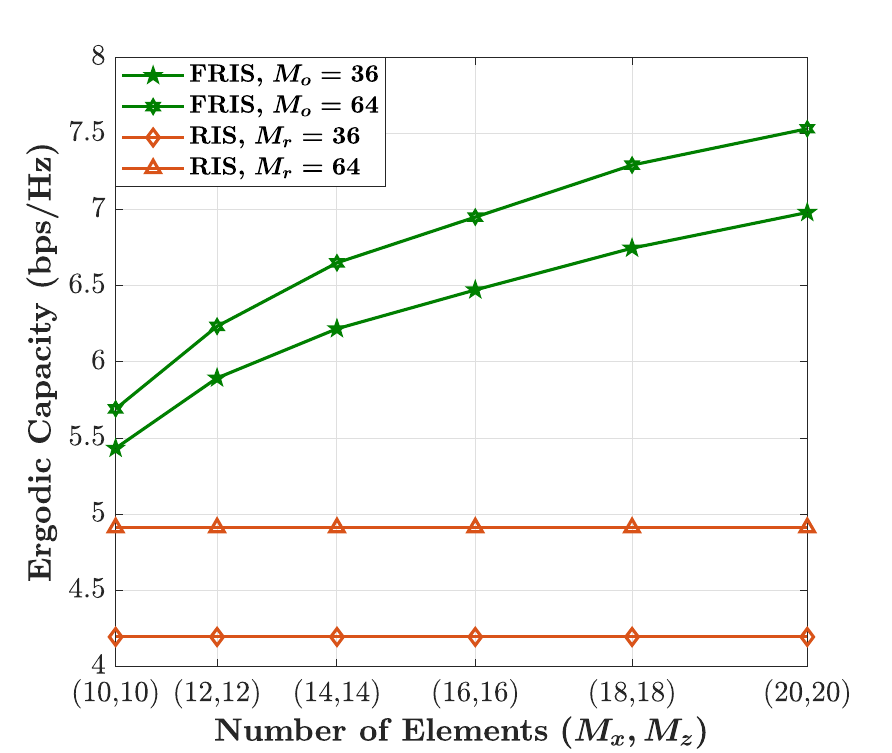}\label{fig_cm}
	}\vspace{-2mm}
\caption{(a) Outage probability performance versus the average SNR $\overline{\gamma}$ for different $M_o$, (b) ergodic capacity performance versus the average SNR $\overline{\gamma}$ for different $M_o$, and (c) ergodic capacity performance versus the number of elements $M=\left(M_x,M_z\right)$ for different $M_o$.}\vspace{-5mm}
\label{}
\end{figure*}

\vspace{-2mm}
\section{Numerical Results and Discussion}\label{sec-num}
In this section, we evaluate the validity of the derived analytical expressions for the outage probability and ergodic capacity of the proposed FRIS-assisted communication system by means of Monte-Carlo simulations. Unless otherwise specified, the simulation parameters are set to: $R = 0.1$, $\rho = 10$, $d_f = 20$ m, $d_u = 40$ m, $\alpha=2.1$, $\left(M_x,M_z\right) = (20,20)$, $W=3\lambda\times 3\lambda$, and a carrier frequency $f_c$ at $2.4$ GHz, hence, $\lambda=0.125$ m. We benchmark the performance against a conventional RIS system of the same size as the FRIS, where $M_r$ denotes the number of its elements.

Fig.~\ref{fig_op} shows the outage probability performance against the average SNR  $\overline{\gamma}$ for different $M_o$ in the proposed FRIS architecture. We compare the results with a conventional RIS with an equal number of reflecting elements, i.e., $M_r \in \{36, 144, 324\}$. The asymptotic analytical results are also included to validate the theoretical accuracy. It can be observed that the asymptotic expressions closely match the Monte-Carlo results in the high-SNR regime, confirming their validity. Moreover, FRIS consistently outperforms the conventional RIS counterpart across all considered values of $M_o$ and $M_r$. Specifically, for a given number of elements, FRIS achieves significantly lower outage probability values compared to the RIS, highlighting its efficiency in enhancing link reliability. 

For instance, at $\overline{\gamma} = 40$ dB with $M_o=M_r=36$, the outage probability achieved by FRIS is close to $10^{-4}$, while RIS only achieves an outage probability above $10^{-3}$. This superior performance with the FRIS is mainly attributed to its capability to dynamically configure the active elements over the aperture, thereby better exploiting spatial diversity. Unlike conventional RISs where elements are fixed and discrete, the FRIS allows for more flexible spatial placement and adaptation, which leads to more favorable propagation conditions and improved signal focusing at the MU. Moreover, as the number of active elements $M_o$ increases, the outage probability improves significantly at first, but tends to saturate beyond a certain point. This saturation behavior arises because the marginal gain from adding more elements diminishes in the high-SNR regime, where the system already benefits from rich spatial diversity. As a result, simply increasing the number of elements does not proportionally enhance performance. 

The ergodic capacity behavior against $\overline{\gamma}$ for both the proposed FRIS and the conventional RIS systems is illustrated in Fig.~\ref{fig_ec}. As shown, in the medium-to-high SNR regime, there is almost perfect agreement between the results from the derived analytical approximation and those from simulations. We also see that FRIS outperforms the conventional RIS across the entire SNR range, with this improvement becoming more pronounced as the SNR increases. This trend demonstrates the FRIS's superior capability in boosting the capacity. For instance, at $\overline{\gamma} = 40$ dB, the FRIS with $M_o = 16$ achieves an ergodic capacity around $11.8$ bps/Hz, while the RIS with the same number of elements only reaches around $8.8$ bps/Hz, i.e., being $34$\% smaller than that obtained with the FRIS. 

Moreover, the results show that increasing the number of elements consistently enhances the ergodic capacity for both architectures. However, as before, the rate of improvement diminishes with larger values of $M_o$ or $M_r$, indicating a saturation effect. This is due to the fact that the system eventually becomes unable to extract additional gains from dense spatial sampling, so further increase in spatial resolution provides marginal capacity benefits. Clearly, FRIS, thanks to its flexible spatial control, is better suited to exploit the available degrees of freedom, thereby improving capacity scaling.

Finally, Fig.~\ref{fig_cm} evaluates the ergodic capacity versus the total number $M=\left(M_x,M_z\right)$ of preset RIS elements. The comparison is made between the FRIS with a fixed number of activated elements $M_o \in \{36, 64\}$ and a conventional RIS with the same total number of reflecting elements $M_r \in \{36, 64\}$. It is evident that the ergodic capacity of FRIS increases consistently with $M$, even when $M_o$ is fixed. This performance gain stems from the FRIS's inherent flexibility, allowing the system to dynamically select the most favorable subset of elements from a denser surface grid. In contrast, the ergodic capacity with RIS remains constant regardless of the total surface resolution. This is because all RIS elements are fixed and passive, thereby, increasing the surface resolution without increasing $M_r$ offers no additional degrees of freedom or optimization capability.

\section{Conclusion}\label{sec-con}
This letter investigated the performance of FRIS-assisted wireless communication systems, presenting a statistical characterization of the equivalent end-to-end channel. Specifically, we derived closed-form expressions for its PDF and CDF using the moment-matching method, providing a fundamental understanding of the channel behavior with an FRIS. We further derived analytical expressions for the outage probability and a tight upper bound for the ergodic capacity, as well as their asymptotic trends in the high-SNR regime. Our numerical results showcased that FRIS significantly improves both reliability and spectral efficiency compared to RIS. These findings underscore the potential of FRIS as a scalable and adaptive solution for future wireless networks.

\end{document}